\documentclass[reqno,10pt,a4paper,dvips]{amsart}

\usepackage{amssymb,mathptmx,cite,psfrag,eucal,array,setspace,geometry,enumitem}
\usepackage[dvips]{graphicx}
\usepackage{tikz}
\numberwithin{equation}{section}

\newcolumntype{C}{>{$}c<{$}} %Defines math mode in tabular (array package)...

\allowdisplaybreaks

\newcommand{\alg}[1]{\mathfrak{#1}}
\newcommand{\group}[1]{\mathsf{#1}}
\newcommand{\uealg}[1]{\mathcal{U} \bigl( #1 \bigr)}

\newcommand{\func}[2]{#1 \left( #2 \right)}
\newcommand{\tfunc}[2]{#1 \bigl( #2 \bigr)}

\newcommand{\brac}[1]{\left( #1 \right)}
\newcommand{\sqbrac}[1]{\left[ #1 \right]}
\newcommand{\set}[1]{\left\{ #1 \right\}}

\newcommand{\abs}[1]{\left| #1 \right|}

\newcommand{\ZZ}{\mathbb{Z}}
\newcommand{\NN}{\mathbb{N}}
\newcommand{\RR}{\mathbb{R}}
\newcommand{\CC}{\mathbb{C}}

\newcommand{\ee}{\mathsf{e}}

\newcommand{\eps}{\varepsilon}

\newcommand{\killing}[2]{\kappa \bigl( #1 , #2 \bigr)}

\newcommand{\affine}[1]{\widehat{#1}}

\newcommand{\comm}[2]{\bigl[ #1 , #2 \bigr]}
\newcommand{\acomm}[2]{\bigl\{ #1 , #2 \bigr\}}

\newcommand{\ket}[1]{\bigl\lvert #1 \bigr\rangle}

\newcommand{\normord}[1]{\mbox{${} : #1 : {}$}} % {} necessary to prevent := or =:

\newcommand{\VerMod}[1]{\mathcal{V}_{#1}}

\newcommand{\ProjMod}[1]{\mathcal{P}_{#1}}
\newcommand{\TypMod}[1]{\mathcal{T}_{#1}}
\newcommand{\AtypMod}[1]{\mathcal{A}_{#1}}

\newcommand{\AffVerMod}[1]{\affine{\mathcal{V}}_{#1}}

\newcommand{\AffProjMod}[1]{\affine{\mathcal{P}}_{#1}}
\newcommand{\AffTypMod}[1]{\affine{\mathcal{T}}_{#1}}
\newcommand{\AffAtypMod}[1]{\affine{\mathcal{A}}_{#1}}

\newcommand{\SLA}[2]{\alg{#1} \left( #2 \right)}
\newcommand{\SLSA}[3]{\alg{#1} \left( #2 \middle\vert #3 \right)}
\newcommand{\AKMA}[2]{\affine{\alg{#1}} \left( #2 \right)}
\newcommand{\AKMSA}[3]{\affine{\alg{#1}} \left( #2 \middle\vert #3 \right)}

\newcommand{\SLSG}[3]{\group{#1} \left( #2 \middle\vert #3 \right)}

\newcommand{\ch}[2]{\tfunc{\chi_{\raisebox{-3pt}{$\scriptstyle #1$}}}{#2}}

\newcommand{\fuse}{\mathbin{\times}}

\newcommand{\Vertop}[1]{\ee^{#1}}
\newcommand{\vertop}[1]{\normord{\Vertop{#1}}}

\newcommand{\dses}[3]{0 \longrightarrow #1 \longrightarrow #2 \longrightarrow #3 \longrightarrow 0}

\newcommand{\eqnref}[1]{Equation~\eqref{#1}}

\newcommand{\eqnTref}[3]{Equations~\eqref{#1}, \eqref{#2} and \eqref{#3}}

\newcommand{\secref}[1]{Section~\ref{#1}}

\newcommand{\propref}[1]{Proposition~\ref{#1}}

\newcommand{\cfts}{conformal field theories}

\newcommand{\lcft}{logarithmic conformal field theory}

\newcommand{\WZW}{Wess-Zumino-Witten}
\newcommand{\ope}{operator product expansion}
\newcommand{\opes}{operator product expansions}
\newcommand{\hws}{highest weight state}
\newcommand{\hwss}{highest weight states}

\newcommand{\hwm}{highest weight module}
\newcommand{\hwms}{highest weight modules}

\newcommand{\bb}{\beta}
\newcommand{\bg}{\gamma}
\newcommand{\bc}{\varkappa}

\newcommand{\BH}{\mathbf{H}}
\newcommand{\BZ}{\mathbf{Z}}
\newcommand{\BE}{\mathbf{E}}
\newcommand{\BF}{\mathbf{F}}
\newcommand{\Be}{\mathbf{e}}
\newcommand{\Bf}{\mathbf{f}}
\newcommand{\sg}{\mathsf{g}}
\newcommand{\sj}{\mathsf{j}}
\newcommand{\st}{\mathsf{t}}

\newcommand{\CE}{\mathcal{E}}
\newcommand{\CH}{\mathcal{H}}
\newcommand{\CT}{\mathcal{T}}
\newcommand{\CW}{\mathcal{W}}

\DeclareMathOperator{\tr}{tr}

\DeclareMathOperator{\diag}{diag}

\newcommand{\traceover}[1]{\tr_{\raisebox{-3pt}{$\scriptstyle #1$}}}

\theoremstyle{plain}
\newtheorem{theorem}{Theorem}%[section]

\newtheorem{proposition}[theorem]{Proposition}

\newtheorem*{conjecture}{Conjecture}

\begin{document}

\title{W-algebras extending $\AKMSA{gl}{1}{1}$}

\author[T Creutzig]{Thomas Creutzig}

\address[T Creutzig]{
Fachbereich Mathematik \\
Technische Universit\"{a}t Darmstadt\\
Schlo\ss{}gartenstra\ss{}e 7\\
64289 Darmstadt\\ Germany
}

\email{tcreutzig@mathematik.tu-darmstadt.de}

\author[D Ridout]{David Ridout}

\address[David Ridout]{
Department of Theoretical Physics \\
Research School of Physics and Engineering;
and
Mathematical Sciences Institute;
Australian National University \\
Canberra, ACT 0200 \\
Australia
}

\email{david.ridout@anu.edu.au}

\thanks{\today}

\begin{abstract}
It was recently shown that $\AKMSA{gl}{1}{1}$ admits an infinite family of simple current extensions.  Here, these findings are reviewed and explicit free field realisations of the extended algebras are constructed.  The leading contributions to the operator product algebra are then calculated.  Among these extensions, one finds four infinite families that seem to contain, as subalgebras, copies of the $W^{\brac{2}}_N$ algebras of Feigin and Semikhatov at various levels and central charges $\pm 1$.
\end{abstract}

\maketitle

\onehalfspacing

\section{Introduction} \label{secIntro}

The affine Kac-Moody superalgebra $\AKMSA{gl}{1}{1}$ is an attractive candidate for study.  On the one hand, its highest weight theory is particularly easy to analyse.  On the other, one is naturally led to study indecomposable modules of the type that arise in \lcft{}.  In \cite{CR:GL11}, we reviewed and consolidated what was known about this superalgebra, drawing in particular upon the previous works \cite{RozQua92,Rozansky:1992td,SalGL106,Creutzig:2007jy,CS09,Creutzig:2009zz,Creutzig:2010ne}.

One motivation for undertaking this work was to understand how one could reconcile the observation that \cfts{} with $\AKMSA{gl}{1}{1}$ symmetry appeared to admit only continuous spectra, whereas one might expect that the \WZW{} model on the real form $\SLSG{U}{1}{1}$ would have the same symmetry, but a discrete spectrum.  Another was to understand whether $\AKMSA{gl}{1}{1}$ could be related to other infinite-dimensional algebras, thus providing relationships between certain (logarithmic) \cfts{}.  For the first question, we were able to show that certain discrete spectra seem to be consistent provided one \emph{extends} the chiral algebra appropriately.  For the second, we identified a certain $\AKMA{u}{1}$-coset of $\AKMSA{gl}{1}{1}$ as the chiral algebra of the well-known $\beta \gamma$ ghost system.  Previous work \cite{RidSL210} then links $\AKMSA{gl}{1}{1}$ to the affine Kac-Moody algebra $\AKMA{sl}{2}_{-1/2}$ \cite{RidSL208,RidFus10}, the triplet algebra $\func{\alg{W}}{1,2}$ of Gaberdiel and Kausch \cite{GabRat96} and the symplectic fermions algebra \cite{KauSym00} ($\AKMSA{psl}{1}{1}$).

This article describes a certain family of \emph{extended algebras} of $\AKMSA{gl}{1}{1}$.  In \cite{CR:GL11}, we noted that the fusion rules give rise to an infinite family of simple currents labelled by $n \in \RR$ and $\ell \in \ZZ$.  It follows that these algebra extensions may be computed algorithmically \cite{RidSU206,RidMin07}.  Here, we perform the computations up to a certain order, using a well-known free field realisation \cite{Guruswamy:1999hi}.  More precisely, we study the resulting W-algebras and show that, for certain infinite families of $n$ and $\ell$, there is a bosonic subalgebra which we conjecture to be the $W^{\brac{2}}_N$ algebra of Feigin and Semikhatov \cite{Feigin:2004wb}.

\section{$\SLSA{gl}{1}{1}$ and its representations} \label{secFinAlg}

\subsection{Algebraic Structure}

The Lie superalgebra $\SLSA{gl}{1}{1}$ consists of the endomorphisms of the super vector space $\CC^{1 \mid 1}$ equipped with the standard graded commutator.  It is convenient to choose the following basis,
\begin{equation} \label{eqngl11DefRep}
N = \frac{1}{2} 
\begin{pmatrix}
1 & 0 \\
0 & -1
\end{pmatrix}
, \qquad E = 
\begin{pmatrix}
1 & 0 \\
0 & 1
\end{pmatrix}
, \qquad \psi^+ = 
\begin{pmatrix}
0 & 1 \\
0 & 0
\end{pmatrix}
, \qquad \psi^- = 
\begin{pmatrix}
0 & 0 \\
1 & 0
\end{pmatrix}
,
\end{equation}
in which $N$ and $E$ are parity-preserving (bosonic) whereas $\psi^+$ and $\psi^-$ are parity-reversing (fermionic).  The non-vanishing brackets are then
\begin{equation} \label{eqngl11Rels}
\comm{N}{\psi^{\pm}} = \pm \psi^{\pm}, \qquad \acomm{\psi^+}{\psi^-} = E.
\end{equation}
We note that $E$ is central, so this superalgebra is not simple.  In fact, $\SLSA{gl}{1}{1}$ does not decompose as a direct sum of ideals.  Equivalently, the adjoint representation of $\SLSA{gl}{1}{1}$ is reducible, but indecomposable.

The standard non-degenerate bilinear form $\killing{\cdot}{\cdot}$ on $\SLSA{gl}{1}{1}$ is given by the supertrace of the product in the defining representation \eqref{eqngl11DefRep}.  With respect to the basis elements \eqref{eqngl11DefRep}, this form is
\begin{equation}
\killing{N}{E} = \killing{E}{N} = 1, \qquad \killing{\psi^+}{\psi^-} = -\killing{\psi^-}{\psi^+} = 1,
\end{equation}
with all other combinations vanishing.  From this, we compute the quadratic Casimir $Q \in \uealg{\SLSA{gl}{1}{1}}$ (up to an arbitrary polynomial in the central element $E$).  We find it convenient to take
\begin{equation} \label{eqnDefCasimir}
Q = NE + \psi^- \psi^+.
\end{equation}

\subsection{Representation Theory} \label{secFinRep}

The obvious triangular decomposition of $\SLSA{gl}{1}{1}$ regards $\psi^+$ as a raising (annihilation) operator, $\psi^-$ as a lowering (creation) operator, and $N$ and $E$ as Cartan elements.  A \hws{} of a $\SLSA{gl}{1}{1}$-representation is then defined to be an eigenstate of $N$ and $E$ which is annihilated by $\psi^+$.  Such states generate Verma modules in the usual way and as $\psi^-$ squares to zero in any representation, every Verma module has dimension $2$.  If $\brac{n,e}$ denotes the weight (the $N$- and $E$-eigenvalues) of a \hws{} generating a Verma module, then its unique descendant will have weight $\brac{n-1,e}$.  We will denote this Verma module by $\VerMod{n-1/2,e}$, remarking that the convention of characterising a \hwm{} by the \emph{average} $N$-eigenvalue of its states, rather than that of the \hws{} itself, turns out to symmetrise many of the formulae to follow.

Suppose now that $\ket{v}$ is a (generating) \hws{} of $\VerMod{n,e}$.  It satisfies
\begin{equation}
\psi^+ \psi^- \ket{v} = \acomm{\psi^+}{\psi^-} \ket{v} = E \ket{v} = e \ket{v},
\end{equation}
so the descendant $\psi^- \ket{v} \neq 0$ is a singular vector if and only if $e=0$.  Verma modules are therefore irreducible for $e \neq 0$, and have irreducible quotients of dimension $1$ when $e = 0$.  Modules with $e \neq 0$ are called \emph{typical} while those with $e = 0$ are \emph{atypical}.  We will denote a typical irreducible by $\TypMod{n,e} \cong \VerMod{n,e}$ and an atypical irreducible by $\AtypMod{n}$.  Our convention of labelling modules by their average $N$-eigenvalue leads us to define the latter to be the irreducible quotient of $\VerMod{n-1/2,0}$.  This is summarised in the short exact sequence
\begin{equation} \label{ESFinV}
\dses{\AtypMod{n-1/2}}{\VerMod{n,0}}{\AtypMod{n+1/2}}
\end{equation}
and structure diagram
\begin{equation}
\parbox[c]{0.4\textwidth}{
\begin{tikzpicture}[auto,thick,
	nom/.style={circle,draw=black!20,fill=black!20,inner sep=2pt}
	]
\node (q1) at (0,0) {$\AtypMod{n+1/2}$};
\node (s1) at (3,0) {$\AtypMod{n-1/2}$};
\node at (-2,0) [nom] {$\VerMod{n,0}$};
\node at (-1.25,0) {$:$};
\draw [->] (q1) to node {$\psi^-$} (s1);
\end{tikzpicture}
} \ .
\end{equation}
Such diagrams illustrate how the irreducible composition factors of a module are combined, with arrows indicating (schematically) the action of the algebra.

Atypical modules also appear as submodules of larger indecomposable modules.  Of particular importance are the four-dimensional projectives\footnote{We mention that the typical irreducibles are also projective in the category of finite-dimensional $\SLSA{gl}{1}{1}$-modules.} $\ProjMod{n}$ whose structure diagrams take the form
\begin{equation} \label{picStaggered}
\parbox[c]{0.28\textwidth}{
\begin{center}
\begin{tikzpicture}[auto,thick,
	nom/.style={circle,draw=black!20,fill=black!20,inner sep=2pt}
	]
\node (top) at (0,1.5) [] {$\AtypMod{n}$};
\node (left) at (-1.5,0) [] {$\AtypMod{n+1}$};
\node (right) at (1.5,0) [] {$\AtypMod{n-1}$};
\node (bot) at (0,-1.5) [] {$\AtypMod{n}$};
\node at (0,0) [nom] {$\ProjMod{n}$};
\draw [->] (top) to node [swap] {$\psi^+$} (left);
\draw [->] (top) to node {$\psi^-$} (right);
\draw [->] (left) to node [swap] {$\psi^-$} (bot);
\draw [->] (right) to node {$-\psi^+$} (bot);
\end{tikzpicture}
\end{center}
}
.
\end{equation}
We remark that these modules may be viewed as particularly simple examples of staggered modules \cite{RidSta09}.  Indeed, they may be regarded as extensions of \hwms{} via the exact sequence
\begin{equation} \label{ESFinP}
\dses{\VerMod{n+1/2,0}}{\ProjMod{n}}{\VerMod{n-1/2,0}},
\end{equation}
and one can verify that the Casimir $Q$ acts non-diagonalisably on $\ProjMod{n}$, taking the generator associated with the top $\AtypMod{n}$ factor to the generator of the bottom $\AtypMod{n}$ factor, while annihilating the other states.

\subsection{The Representation Ring}

The relevance of the projectives $\ProjMod{n}$ is that they appear in the representation ring generated by the irreducibles.\footnote{It is perhaps also worth pointing out that the adjoint representation of $\SLSA{gl}{1}{1}$ is isomorphic to $\ProjMod{0}$.}  The tensor product rules governing this ring are \cite{RozQua92}
\begin{equation} \label{RepRing}
\begin{gathered}
\AtypMod{n} \otimes \AtypMod{n'} = \AtypMod{n+n'}, \qquad
\AtypMod{n} \otimes \TypMod{n',e'} = \TypMod{n+n',e'}, \qquad
\AtypMod{n} \otimes \ProjMod{n'} = \ProjMod{n+n'}, \\
\TypMod{n,e} \otimes \TypMod{n',e'} = 
\begin{cases}
\ProjMod{n+n'} & \text{if $e+e'=0$,} \\
\TypMod{n+n'+1/2,e+e'} \oplus \TypMod{n+n'-1/2,e+e'} & \text{otherwise,}
\end{cases}
\\
\TypMod{n,e} \otimes \ProjMod{n'} = \TypMod{n+n'+1,e} \oplus 2 \: \TypMod{n+n',e} \oplus \TypMod{n+n'-1,e}, \qquad
\ProjMod{n} \otimes \ProjMod{n'} = \ProjMod{n+n'+1} \oplus 2 \: \ProjMod{n+n'} \oplus \ProjMod{n+n'-1}.
\end{gathered}
\end{equation}
There are other indecomposables which may be constructed from submodules and quotients of the $\ProjMod{n}$ by taking tensor products.  We will not need them and refer to \cite{GotRep07} for further discussion.

\section{$\AKMSA{gl}{1}{1}$ and its Representations} \label{secAffine}

\subsection{Algebraic Structure} \label{secAffAlg}

Our conventions for $\SLSA{gl}{1}{1}$ carry over to its affinisation $\AKMSA{gl}{1}{1}$ in the usual way.  Explicitly, the non-vanishing brackets are
\begin{equation}
\comm{N_r}{E_s} = r k \delta_{r+s,0}, \qquad \comm{N_r}{\psi^{\pm}_s} = \pm \psi^{\pm}_{r+s}, \qquad \acomm{\psi^+_r}{\psi^-_s} = E_{r+s} + r k \delta_{r+s,0},
\end{equation}
where $k \in \RR$ is called the level and $r,s \in \ZZ$.  We emphasise that when $k \neq 0$, the generators can be rescaled so as to normalise $k$ to $1$:
\begin{equation} \label{eqnLevelScaling}
N_r \longrightarrow N_r, \qquad E_r \longrightarrow \frac{E_r}{k}, \qquad \psi^{\pm}_r \longrightarrow \frac{\psi^{\pm}_r}{\sqrt{k}}.
\end{equation}
As in the more familiar case of $\AKMA{u}{1}$, we see that the actual value of $k \neq 0$ is not physical.

The Virasoro generators are constructed using (a modification of) the Sugawara construction.  Because the quadratic Casimir of $\SLSA{gl}{1}{1}$ is only defined modulo polynomials in $E$, one tries the ansatz \cite{RozQua92}
\begin{equation} \label{eqnDefT}
\func{T}{z} = \mu \func{\normord{NE + EN - \psi^+ \psi^- + \psi^- \psi^+}}{z} + \nu \func{\normord{EE}}{z},
\end{equation}
finding that this defines an energy-momentum tensor if and only if $\mu = 1/2k$ and $\nu = 1/2k^2$.  Moreover, the $\AKMSA{gl}{1}{1}$ currents $\func{N}{z}$, $\func{E}{z}$ and $\func{\psi^{\pm}}{z}$ are found to be Virasoro primaries of conformal dimension $1$ and the central charge is zero.

The structure theory of \hwms{} for $\AKMSA{gl}{1}{1}$ turns out to be particularly accessible because of certain automorphisms.  These consist of the automorphism $\mathsf{w}$ which defines the notion of conjugation and the family \cite{SalGL106} of spectral flow automorphisms $\sigma^{\ell}$, $\ell \in \ZZ$.  Explicitly,
\begin{equation}
\begin{aligned}
\func{\mathsf{w}}{N_r} &= -N_{r}, \\
\func{\sigma^{\ell}}{N_r} &= N_r,
\end{aligned}
\qquad
\begin{aligned}
\func{\mathsf{w}}{E_r} &= -E_{r}, \\
\func{\sigma^{\ell}}{E_r} &= E_r - \ell k \delta_{r,0},
\end{aligned}
\qquad
\begin{aligned}
\func{\mathsf{w}}{\psi^{\pm}_r} &= \pm \psi^{\mp}_{r}, \\
\func{\sigma^{\ell}}{\psi^{\pm}_r} &= \psi^{\pm}_{r \mp \ell},
\end{aligned}
\qquad
\begin{aligned}
\func{\mathsf{w}}{L_0} &= L_0. \\
\func{\sigma^{\ell}}{L_0} &= L_0 - \ell N_0.
\end{aligned}
\end{equation}
These automorphisms may be used to construct new modules $\func{\mathsf{w}^*}{\mathcal{M}}$ and $\func{\sigma^*}{\mathcal{M}}$ by twisting the action of the algebra on a module $\mathcal{M}$:
\begin{equation} \label{eqnInducedAction}
J \cdot \tfunc{\mathsf{w}^*}{\ket{v}} = \func{\mathsf{w}^*}{\tfunc{\mathsf{w}^{-1}}{J} \ket{v}}, \qquad J \cdot \tfunc{\sigma^*}{\ket{v}} = \func{\sigma^*}{\tfunc{\sigma^{-1}}{J} \ket{v}} \qquad \text{($J \in \AKMSA{gl}{1}{1}$).}
\end{equation}
Note that $\func{\mathsf{w}^*}{\mathcal{M}}$ is precisely the module conjugate to $\mathcal{M}$.

\subsection{Representation Theory} \label{secAffRep}

We can now define affine \hwss{}, affine Verma modules $\AffVerMod{n,\ell}$, and their irreducible quotients as before.  We remark only that \eqref{eqnLevelScaling} suggests that we characterise modules by the invariant ratio $\ell = e/k$ rather than by the $E_0$-eigenvalue $e$.  The affine \hws{} $\ket{v_{n,\ell}}$ of $\AffVerMod{n,\ell}$, whose weight (its $N_0$- and $E_0/k$-eigenvalues) is $\brac{n + \tfrac{1}{2}, \ell}$, has conformal dimension
\begin{equation} \label{eqnConfDim}
\Delta_{n,\ell} = n \ell + \frac{1}{2} \ell^2.
\end{equation}
Of course, this formula also applies to singular vectors.  Again, the label $n$ refers to the average $N_0$-eigenvalue of the zero-grade subspace of $\AffVerMod{n,\ell}$, generalising the labelling convention of \secref{secFinRep}.

Verma modules for $\AKMSA{gl}{1}{1}$ are infinite-dimensional and their characters have the form
\begin{equation} \label{eqnCharVerma}
\ch{\AffVerMod{n,\ell}}{z;q} = \traceover{\AffVerMod{n,\ell}} z^{N_0} q^{L_0} = z^{n+1/2} q^{\Delta_{n,\ell}} \prod_{i=1}^{\infty} \frac{\brac{1 + z q^i} \brac{1 + z^{-1} q^{i-1}}}{\brac{1 - q^i}^2}.
\end{equation}
For the irreducible quotients, the case with $\ell = 0$ is particularly easy.  As in \secref{secFinRep}, we regard $\brac{n,\ell}$ (and modules so-labelled) as being \emph{typical} if $\AffVerMod{n,\ell}$ is irreducible and \emph{atypical} otherwise.
\begin{proposition} \label{prop:ell=0}
The affine Verma module $\AffVerMod{n,0}$ has an exact sequence
\begin{equation}
\dses{\AffAtypMod{n-1/2,0}}{\AffVerMod{n,0}}{\AffAtypMod{n+1/2,0}}
\end{equation}
in which the $\AffAtypMod{n,0}$ are (atypical) irreducibles whose characters are given by
\begin{equation} \label{eqnCharVac}
\ch{\AffAtypMod{n,0}}{z;q} = z^n \prod_{i=1}^{\infty} \frac{\brac{1 + z q^i} \brac{1 + z^{-1} q^i}}{\brac{1 - q^i}^2}.
\end{equation}
\end{proposition}
\begin{proof}
Since $\ell = 0$, every singular vector of $\AffVerMod{n,0}$ has dimension $0$ by \eqnref{eqnConfDim}.  The space of singular vectors is thus spanned by $\ket{v_{n,0}}$ and $\psi^-_0 \ket{v_{n,0}}$.  Taking the quotient by the module generated by $\psi^-_0 \ket{v_{n,0}}$ gives a module with a one-dimensional zero-grade subspace.  The only singular vector is then the \hws{}, so this quotient is irreducible.  We denote it by $\AffAtypMod{n+1/2,0}$ as its zero-grade subspace has $N_0$-eigenvalue $n+\tfrac{1}{2}$.  Its character follows trivially.  The submodule of $\AffVerMod{n,0}$ generated by $\psi^-_0 \ket{v_{n,0}}$ is not a Verma module because $\bigl( \psi^-_0 \bigr)^2 \ket{v_{n,0}} = 0$.  It must therefore be a proper quotient of $\AffVerMod{n-1,0}$ and, by the above argument, the only such quotient is the irreducible $\AffAtypMod{n-1/2,0}$.  The exact sequence follows.
\end{proof}
For $\ell \neq 0$, one proves by direct calculation \cite{CR:GL11} that for $0 < \abs{\ell} < 1$, $\AffVerMod{n,\ell}$ is irreducible.  In other words, the corresponding irreducibles are typical, hence we denote them by $\AffTypMod{n,\ell}$.  For $\abs{\ell} \geqslant 1$, the structure of the Verma modules now follows from considering the induced action of the spectral flow automorphisms.  More precisely, one proves \cite{CR:GL11} that any Verma module is isomorphic to a twisted version of a Verma module with $-1 < \abs{\ell} < 1$ (or the conjugate of such a Verma module).  We summarise the result as follows.
\begin{proposition}
When $\ell \notin \ZZ$, the affine Verma module $\AffVerMod{n,\ell}$ is irreducible, $\AffVerMod{n,\ell} \cong \AffTypMod{n,\ell}$, so its character is given by \eqnref{eqnCharVerma}.  When $\ell \in \ZZ$, the affine Verma module $\AffVerMod{n,\ell}$ has an exact sequence
\begin{equation}
\begin{gathered}
\dses{\AffAtypMod{n+1,\ell}}{\AffVerMod{n,\ell}}{\AffAtypMod{n,\ell}} \qquad \text{($\ell = +1,+2,+3,\ldots$),} \\
\dses{\AffAtypMod{n-1,\ell}}{\AffVerMod{n,\ell}}{\AffAtypMod{n,\ell}} \qquad \text{($\ell = -1,-2,-3,\ldots$),}
\end{gathered}
\end{equation}
in which the $\AffAtypMod{n,\ell}$ are (atypical) irreducibles whose characters are given by
\begin{equation} \label{eqnCharAtyp}
\ch{\AffAtypMod{n,\ell}}{z;q} = 
\begin{cases}
\displaystyle \frac{z^{n+1/2} q^{\Delta_{n,\ell}}}{1 + zq^{\ell}} \prod_{i=1}^{\infty} \frac{\brac{1 + z q^i} \brac{1 + z^{-1} q^{i-1}}}{\brac{1 - q^i}^2} & \text{($\ell = +1,+2,+3,\ldots$),} \\
\displaystyle \frac{z^{n+1/2} q^{\Delta_{n,\ell}}}{1 + z^{-1} q^{-\ell}} \prod_{i=1}^{\infty} \frac{\brac{1 + z q^i} \brac{1 + z^{-1} q^{i-1}}}{\brac{1 - q^i}^2} & \text{($\ell = -1,-2,-3,\ldots$).}
\end{cases}
\end{equation}
(The exact sequence and character for $\ell = 0$ was given in \propref{prop:ell=0}.)
\end{proposition}
\noindent Note that the $\AffVerMod{n,\ell}$ with $\ell \in \ZZ$ have a non-trivial singular vector at grade $\abs{\ell}$.  We emphasise that the $\AffAtypMod{n,\ell}$ with $\ell \neq 0$ therefore possess a two-dimensional zero-grade subspace.

This description of the Verma modules, their irreducible quotients and characters relies upon being able to identify the result of applying the spectral flow automorphisms to modules.  For irreducibles, we have
\begin{equation}
\tfunc{\bigl( \sigma^{\ell'} \bigr)^*}{\AffTypMod{n,\ell}} = \AffTypMod{n-\ell',\ell+\ell'}, \qquad \tfunc{\bigl( \sigma^{\ell'} \bigr)^*}{\AffAtypMod{n,\ell}} = \AffAtypMod{n-\ell'+\func{\eps}{\ell+\ell'}-\func{\eps}{\ell},\ell+\ell'},
\end{equation}
where we introduce a convenient variant $\eps$ of the sign function on $\ZZ$, defined by taking $\func{\eps}{\ell}$ to be $\tfrac{1}{2}$, $0$ or $-\tfrac{1}{2}$ according as to whether $\ell \in \ZZ$ is positive, zero or negative, respectively.

\subsection{Fusion} \label{secAffFus}

The fusion rules of the irreducible $\AKMSA{gl}{1}{1}$-modules (among others) were first deduced in \cite{Creutzig:2007jy} using three-point functions computed in a free field realisation and a conjectured completeness of the spectrum.  These rules and the spectrum conjecture were confirmed in \cite{CR:GL11} through a direct argument involving the Nahm-Gaberdiel-Kausch fusion algorithm \cite{NahQua94,GabInd96} and spectral flow.  The fusion ring generated by the irreducibles may be understood \cite{Quella:2007hr} as a ``constrained lift'' of the representation ring \eqref{RepRing} of $\SLSA{gl}{1}{1}$ where the constraints are effectively implemented by spectral flow.  Explicitly, the rules are
\begin{equation} \label{Fusion}
\begin{gathered}
\AffAtypMod{n,\ell} \fuse \AffAtypMod{n',\ell'} = \AffAtypMod{n+n'-\func{\eps}{\ell,\ell'},\ell+\ell'}, \quad
\AffAtypMod{n,\ell} \fuse \AffTypMod{n',\ell'} = \AffTypMod{n+n'-\func{\eps}{\ell},\ell+\ell'}, \quad
\AffAtypMod{n,\ell} \fuse \AffProjMod{n',\ell'} = \AffProjMod{n+n'-\func{\eps}{\ell,\ell'},\ell+\ell'}, \\
\AffTypMod{n,\ell} \fuse \AffTypMod{n',\ell'} = 
\begin{cases}
\AffProjMod{n+n'+\func{\eps}{\ell+\ell'},\ell+\ell'} & \text{if $\ell+\ell'=0$,} \\
\AffTypMod{n+n'+1/2,\ell+\ell'} \oplus \AffTypMod{n+n'-1/2,\ell+\ell'} & \text{otherwise,}
\end{cases}
\\
\AffTypMod{n,\ell} \fuse \AffProjMod{n',\ell'} = \AffTypMod{n+n'+1-\func{\eps}{\ell'},\ell+\ell'} \oplus 2 \: \AffTypMod{n+n'-\func{\eps}{\ell'},\ell+\ell'} \oplus \AffTypMod{n+n'-1-\func{\eps}{\ell'},\ell+\ell'}, \\
\AffProjMod{n,\ell} \fuse \AffProjMod{n',\ell'} = \AffProjMod{n+n'+1-\func{\eps}{\ell,\ell'},\ell+\ell'} \oplus 2 \: \AffProjMod{n+n'-\func{\eps}{\ell,\ell'},\ell+\ell'} \oplus \AffProjMod{n+n'-1-\func{\eps}{\ell,\ell'},\ell+\ell'}.
\end{gathered}
\end{equation}
Here, we have defined $\func{\eps}{\ell , \ell'} = \func{\eps}{\ell} + \func{\eps}{\ell'} - \func{\eps}{\ell + \ell'}$ for convenience.

These fusion rules also introduce the indecomposable modules $\AffProjMod{n,\ell}$ which are the counterparts of the projective $\SLSA{gl}{1}{1}$-modules $\ProjMod{n}$ discussed in \secref{secFinRep}.\footnote{More precisely, $\AffProjMod{n,0}$ is the affine counterpart to $\ProjMod{n}$ and the remaining $\AffProjMod{n,\ell}$ are obtained by spectral flow.}  The $\AffProjMod{n,\ell}$ are staggered with structure diagram
\begin{equation} \label{picAffineStaggered}
\parbox[c]{0.28\textwidth}{
\begin{center}
\begin{tikzpicture}[auto,thick,
	nom/.style={circle,draw=black!20,fill=black!20,inner sep=2pt}
	]
\node (top) at (0,1.5) [] {$\AffAtypMod{n,\ell}$};
\node (left) at (-1.5,0) [] {$\AffAtypMod{n+1,\ell}$};
\node (right) at (1.5,0) [] {$\AffAtypMod{n-1,\ell}$};
\node (bot) at (0,-1.5) [] {$\AffAtypMod{n,\ell}$};
\node at (0,0) [nom] {$\AffProjMod{n,\ell}$};
\draw [->] (top) to (left);
\draw [->] (top) to (right);
\draw [->] (left) to (bot);
\draw [->] (right) to (bot);
\end{tikzpicture}
\end{center}
}
\end{equation}
and a non-diagonalisable action of the Virasoro mode $L_0$.  It follows that \cfts{} whose spectra contain typical modules will also contain such $\AffProjMod{n,\ell}$ (by fusion), and so will be \emph{logarithmic}.

\section{W-Algebras extending $\AKMSA{gl}{1}{1}$} \label{secExtAlg}

\subsection{Chiral Algebra Extensions}

Our search for extended algebras is guided by the following considerations:  First, note that if we choose to extend by a zero-grade field associated to any irreducible $\AKMSA{gl}{1}{1}$-module, then we must include the rest of its zero-grade fields in the extension.  Second, the fields we extend by should be closed under conjugation.  Third, extending by fields from typical irreducibles will lead to logarithmic behaviour in the extended chiral algebra because fusing typicals with their conjugates yields the staggered indecomposable $\AffProjMod{0,0}$.

It seems then that the most tractable extensions will involve zero-grade fields from atypical modules $\AffAtypMod{n,\ell}$ and their conjugates $\AffAtypMod{-n,-\ell}$.  The simplest extension we could hope for would involve a single atypical and its conjugate and have the further property that these extension fields generate no new fields at the level of the commutation relations.  This may be achieved for extension fields of integer or half-integer conformal dimension by requiring that the \opes{} of the zero-grade fields of $\AffAtypMod{n,\ell}$ are regular.  From the fusion rules \eqref{Fusion}, we obtain
\begin{equation}
\AffAtypMod{n,\ell} \fuse \AffAtypMod{n,\ell} = \AffAtypMod{2n - \func{\eps}{\ell},2\ell},
\end{equation}
from which it follows that the zero-grade fields of $\AffAtypMod{n,\ell}$ will have regular \opes{} with one another if $2 \: \Delta_{n,\ell} \leqslant \Delta_{2n - \func{\eps}{\ell},2\ell}$, that is, if
\begin{equation}\label{eqdim}
\abs{\ell} \leqslant 2 \: \Delta_{n,\ell}.
\end{equation}
We may take $\ell$ positive without loss of generality.  Further, we require that the conformal dimension of the extension fields be a positive half-integer (so $2n\ell \in \ZZ$).  \eqnref{eqdim} then implies that there are $m$ distinct possibilities to extend by fields of dimension $m/2$.
We denote by $\alg{W}_{n,\ell}$ the algebra obtained upon extending $\AKMSA{gl}{1}{1}$ by the atypical module $\AffAtypMod{n,\ell}$ and its conjugate $\AffAtypMod{-n,-\ell}$.

\subsection{Characters of Extended Algebras}

The complete extended algebra also contains normally-ordered products of the extension fields and their descendants.  Indeed, the extended algebra $\alg{W}_{n,\ell}$ may be identified, at least at the level of graded vector spaces, with the orbit of the $\AKMSA{gl}{1}{1}$ vacuum module under fusion by the simple current modules $\AffAtypMod{n,\ell}$ and $\AffAtypMod{-n,-\ell}$.  In other words,
\begin{equation}
\alg{W}_{n+1/2,\ell} = \AffAtypMod{0,0} \oplus \bigoplus_{m=1}^{\infty} \bigl( \AffAtypMod{mn+1/2,m\ell} \oplus \AffAtypMod{-mn-1/2,-m\ell} \bigr).
\end{equation}
The character of the extended vacuum module is therefore
\begin{equation} \label{eqnCharW}
\begin{split}
\ch{\alg{W}_{n+1/2,\ell}}{y;z;q} &= \ch{\AffAtypMod{0,0}}{y,z;q} + \sum_{m=1}^\infty \Bigl[ \ch{\AffAtypMod{mn+1/2, m \ell}}{y;z;q} + \ch{\AffAtypMod{-mn-1/2, -m \ell}}{y;z;q} \Bigr] \\
&= z \sum_{m \in \ZZ} \frac{y^{m \ell} z^{mn} q^{\brac{mn+1/2} m \ell + m^2 \ell^2 / 2}}{1 + z q^{m \ell}} \cdot \prod_{i=1}^{\infty} \frac{\brac{1 + z q^i} \brac{1 + z^{-1} q^{i-1}}}{\brac{1 - q^i}^2}.
\end{split}
\end{equation}
Here, we have introduced an additional formal variable $y$ in order to keep track of the eigenvalues of $E_0 / k$.  One can likewise identify the irreducible modules of the extended algebra with the other orbits of the extension modules.  We will not consider these modules, their characters, nor their interesting modular properties here, but will return to this in a future publication.

\subsection{Free Field Realisations} \label{appFreeFields}

The affine Kac-Moody superalgebra $\AKMSA{gl}{1}{1}$ has two well-known free field realizations, the standard Wakimoto realization \cite{SalGL106} and one constructed from a pair of symplectic fermions, a euclidean boson, and a lorentzian boson \cite{Guruswamy:1999hi}.  An explicit equivalence between the two realisations was established in \cite{CR09}.  Here, we review the latter one.

We take the symplectic fermions $\chi^{\pm}$ and bosons $Y$, $Z$ to have the following \opes{}:
\begin{equation}
\func{\chi^+}{z} \func{\chi^-}{w} = \frac{1}{\brac{z-w}^2} + \text{ regular terms}, \qquad 
\func{\partial Y}{z} \func{\partial Z}{w} = \frac{1}{\brac{z-w}^2} + \text{ regular terms}
\end{equation}
(the others are regular).  The $\AKMSA{gl}{1}{1}$ current fields are then given by
\begin{equation} \label{eqnGL11FFR}
\func{E}{z} = k \func{\partial Y}{z}, \qquad \func{N}{z} = \func{\partial Z}{z}, \qquad \func{\psi^{\pm}}{z} = \sqrt{k} \vertop{\pm \func{Y}{z}} \func{\chi^{\pm}}{z},
\end{equation}
and a moderately tedious computation shows that the $\AKMSA{gl}{1}{1}$ energy momentum tensor \eqref{eqnDefT} indeed corresponds to the sum of those of the bosonic and symplectic fermion systems.

It remains to construct the $\AKMSA{gl}{1}{1}$ primaries that generate our extended algebras.  As these correspond to atypical modules, this is relatively straight-forward.  First, we introduce some convenient notation:
 Let $X_{n,\ell}$ be the bosonic linear combination $n Y + \ell Z$ and define composite fields $F^{\pm}_r$, with $r \in \NN$, by $F^{\pm}_0 = 1$ and $F^{\pm}_r = \normord{F^{\pm}_{r-1} \partial^{r-1} \chi^{\pm}}$ for $r \geqslant 1$.  The conformal dimension of $F^{\pm}_r$ is then $\tfrac{1}{2} r \brac{r+1}$.  The zero-grade fields of the atypicals $\AffAtypMod{n,\ell}$ for $\ell > 0$ have conformal dimension $\Delta_{n,\ell} = \ell \brac{n + \ell / 2}$ and are realised by
\begin{equation}
V_{n,\ell}^+ = \vertop{X_{n+1/2,\ell}} F^-_{\ell-1}, \qquad V_{n,\ell}^- = \vertop{X_{n-1/2,\ell}} F^-_{\ell}.
\end{equation}
This follows from their \opes{} with the $\AKMSA{gl}{1}{1}$ currents:
\begin{equation}
\begin{aligned}
\func{N}{z} \func{V_{n,\ell}^{\pm}}{w} &= \frac{\brac{n \pm 1/2} \: \func{V_{n,\ell}^{\pm}}{w}}{z-w} + \ldots , \\
\func{E}{z} \func{V_{n,\ell}^{\pm}}{w} &= \frac{\ell k \: \func{V_{n,\ell}^{\pm}}{w}}{z-w} + \ldots ,
\end{aligned}
\qquad
\begin{aligned}
\func{\psi^+}{z} \func{V_{n,\ell}^-}{w} &= \brac{-1}^{\ell - 1} \ell ! \frac{\sqrt{k} \: \func{V_{n,\ell}^+}{w}}{z-w} + \ldots , \\
\func{\psi^-}{z} \func{V_{n,\ell}^+}{w} &= \frac{(-1)^{\ell-1}}{\brac{\ell-1}!} \frac{\sqrt{k} \: \func{V_{n,\ell}^-}{w}}{z-w} + \ldots ,
\end{aligned}
\end{equation}
the others being regular.  The zero-grade fields of the conjugate module $\AffAtypMod{-n,-\ell}$ are realised as
\begin{equation}
V_{-n,-\ell}^+ = \vertop{X_{-n+1/2,-\ell}} F^+_{\ell}, \qquad V_{-n,-\ell}^- = \vertop{X_{-n-1/2,-\ell}} F^+_{\ell-1}.
\end{equation}
Their \opes{} with the current fields are similar.

\subsection{The Extended Operator Product Algebra}

In order to compute the leading contributions to the extended algebra \opes{}, we need the expansion of the bosonic vertex operators.  To second order, this is
\begin{multline} \label{eqnVertexOPE}
\vertop{\func{X_{n,\ell}}{z}} \vertop{\func{X_{n',\ell'}}{w}} = \brac{z-w}^{n\ell'+n'\ell} \biggl[ \vertop{\func{X_{n+n',\ell+\ell'}}{w}} + \normord{\func{\partial X_{n,\ell}}{w} \Vertop{\func{X_{n+n',\ell+\ell'}}{w}}} \brac{z-w} \Biggr. \\
\Biggl. + \frac{1}{2} \normord{\Bigl( \func{\partial X_{n,\ell}}{w} \func{\partial X_{n,\ell}}{w} + \func{\partial^2 X_{n,\ell}}{w} \Bigr) \Vertop{\func{X_{n+n',\ell+\ell'}}{w}}} \brac{z-w}^2 + \ldots \biggr].
\end{multline}
Note that it follows that $\vertop{\func{X_{n,\ell}}{w}}$ and $\vertop{\func{X_{n',\ell'}}{w}}$ will be mutually bosonic when $n\ell' + n'\ell$ is an even integer and mutually fermionic when $n\ell' + n'\ell$ is odd.  The implication of this for the statistics of the extended algebra generators $V_{n,\ell}^{\pm}$ and  $V_{-n,-\ell}^{\pm}$ is a little subtle.  It turns out that when $2n \ell$ is even, these generators may be consistently assigned a bosonic or fermionic parity --- $\alg{W}_{n,\ell}$ is a superalgebra.  In fact, $V_{n,\ell}^+$ and $V_{-n,-\ell}^-$ will be fermions and $V_{n,\ell}^-$ and $V_{-n,-\ell}^+$ will be bosons in this case.  However, when $2n \ell$ is odd, such an assignment is impossible --- $\alg{W}_{n,\ell}$ is \emph{not} a superalgebra.  In this case, separately taking $V_{n,\ell}^+$ and $V_{-n,-\ell}^-$ to be bosons and $V_{n,\ell}^-$ and $V_{-n,-\ell}^+$ to be fermions is consistent, but the mutual locality of a boson and a fermion will now be $-1$ instead of $+1$.  We will remark further on this subtlety in \secref{secExamples}.

We moreover need the leading terms of certain \opes{} of the $F^{\pm}_r$.  In particular,
\begin{equation} \label{eqnCompositeSFOPEs}
\begin{split}
\func{F^+_r}{z} \func{F^-_r}{w} &= \brac{z-w}^{-r \brac{r+1}} \biggl[ \mu_r^{\brac{0}} + \mu_{r-1}^{\brac{2}} \normord{\func{\chi^+}{w} \func{\chi^-}{w}} \brac{z-w}^2 + \ldots \biggr] , \\
\func{F^-_{r-1}}{z} \func{F^+_r}{w} &= \brac{z-w}^{-\brac{r-1} \brac{r+1}} \biggl[ \mu_{r-1}^{\brac{1}} \: \func{\chi^+}{w} + \ldots \biggr] , \\
\func{F^-_r}{z} \func{F^+_{r-1}}{w} &= \brac{z-w}^{-\brac{r-1} \brac{r+1}} \biggl[ \mu_{r-1}^{\brac{1}} \: \func{\chi^-}{w} + \ldots \biggr] ,
\end{split}
\end{equation}
where the coefficients $\mu_r^{\brac{a}}$, for $a = 0$, $1$, $2$, are given by
\begin{equation} \label{eq:coeff}
\mu_r^{\brac{a}} = \sum_{\sigma \in \group{S}_r} \brac{-1}^{\abs{\sigma}} \prod_{i=1}^r \brac{i + \func{\sigma}{i} + a - 1}! = \prod_{i=1}^r \brac{i-1}! \brac{i+a}!
\end{equation}
This last equality follows from recognising the $\mu_r^{\brac{a}}$ as determinants of Hankel matrices for which LU-decompositions are easily found.  In detail, consider the $r \times r$ matrix $\func{A_r}{a}$, for a non-negative integer $a$, with entries $\brac{\func{A_r}{a}}_{ij} = \brac{i+j+a-1}!$  Defining $r \times r$ matrices $\func{L_r}{a}$ and $\func{U_r}{a}$ by
\begin{equation}
\brac{\func{L_r}{a}}_{ij} = \frac{\brac{i+a}!}{\brac{j+a}!} \binom{i-1}{j-1}, \qquad 
\brac{\func{U_r}{a}}_{ij} = \brac{i-1}! \brac{j+a}! \binom{j-1}{i-1},
\end{equation}
and noting that $\func{L_r}{a}$ is lower-triangular with diagonal entries equal to $1$ and $\func{U_r}{a}$ is upper-triangular, we see that $\func{L_r}{a} \func{U_r}{a}$ is an LU-decomposition of $\func{A_r}{a}$:
\begin{equation}
\begin{split}
\brac{\func{L_r}{a} \func{U_r}{a}}_{ij} &= \sum_{k=1}^r \frac{\brac{i+a}! \brac{i-1}! \brac{j+a}! \brac{j-1}!}{\brac{k+a}! \brac{k-1}! \brac{i-k}! \brac{j-k}!} = \brac{j+a}! \brac{i-1}! \sum_{k=1}^r \binom{i+a}{k+a} \binom{j-1}{k-1} \\
&= \brac{j+a}! \brac{i-1}! \binom{i+j+a-1}{i-1} = \brac{\func{A_r}{a}}_{ij}.
\end{split}
\end{equation}
Since $\det \: \func{L_r}{a} = 1$, we obtain $\det \: \func{A_r}{a} = \det \: \func{U_r}{a} = \prod_{i=1}^r \brac{i-1}! \brac{i+a}!$ and hence \eqnref{eq:coeff}.

We are now in a position to obtain the leading contributions to the
\opes{} of the extension fields $V_{n,\ell}^{\pm}$ and their conjugates $V_{-n,-\ell}^{\mp}$.  Since we assume \eqref{eqdim}, there are only four non-regular expansions and these take the form
\begin{equation} \label{eqnGenExtAlgOPEs}
\begin{split}
\func{V_{n,\ell}^+}{z} \func{V_{-n,-\ell}^+}{w} &= \frac{\mu_{\ell-1}^{\brac{1}} \: \func{\psi^+}{w} / \sqrt{k}}{\brac{z-w}^{2 \Delta_{n,\ell} - 1}} + \ldots , \\
\func{V_{-n,-\ell}^-}{z} \func{V_{n,\ell}^+}{w} &= \mu_{\ell-1}^{\brac{0}} \Biggl[ \frac{1}{\brac{z-w}^{2 \Delta_{n,\ell}}} - \frac{\func{\partial X_{n+1/2,\ell}}{w}}{\brac{z-w}^{2 \Delta_{n,\ell} - 1}} + \frac{\ell \brac{\ell-1}}{2} \frac{\normord{\func{\chi^+}{w} \func{\chi^-}{w}}}{\brac{z-w}^{2 \Delta_{n,\ell} - 2}} \Biggr. \\
& \mspace{90mu} \Biggl. + \frac{1}{2} \frac{\normord{\func{\partial X_{n+1/2,\ell}}{w} \func{\partial X_{n+1/2,\ell}}{w}} - \func{\partial^2 X_{n+1/2,\ell}}{w}}{\brac{z-w}^{2 \Delta_{n,\ell} - 2}} + \ldots \Biggr] , \\
\func{V_{-n,-\ell}^+}{z} \func{V_{n,\ell}^-}{w} &= \mu_{\ell}^{\brac{0}} \Biggl[ \frac{1}{\brac{z-w}^{2 \Delta_{n,\ell}}} - \frac{\func{\partial X_{n-1/2,\ell}}{w}}{\brac{z-w}^{2 \Delta_{n,\ell} - 1}} + \frac{\ell \brac{\ell+1}}{2} \frac{\normord{\func{\chi^+}{w} \func{\chi^-}{w}}}{\brac{z-w}^{2 \Delta_{n,\ell} - 2}} \Biggr. \\
& \mspace{90mu} \Biggl. + \frac{1}{2} \frac{\normord{\func{\partial X_{n-1/2,\ell}}{w} \func{\partial X_{n-1/2,\ell}}{w}} - \func{\partial^2 X_{n-1/2,\ell}}{w}}{\brac{z-w}^{2 \Delta_{n,\ell} - 2}} + \ldots \Biggr] , \\
\func{V_{n,\ell}^-}{z} \func{V_{-n,-\ell}^-}{w} &= \frac{\mu_{\ell-1}^{\brac{1}} \: \func{\psi^-}{w} / \sqrt{k}}{\brac{z-w}^{2 \Delta_{n,\ell} - 1}} + \ldots
\end{split}
\end{equation}
Here, we have used \eqref{eq:coeff} to evaluate the ratios $\mu_{r-1}^{\brac{2}} / \mu_r^{\brac{0}} = \tfrac{1}{2} r \brac{r+1}$ appearing in these expansions.

\subsection{Examples} \label{secExamples}

Let us now illustrate the results of the above calculations with a few simple examples.  First, \eqref{eqdim} tells us that the extended algebra $\alg{W}_{n,\ell}$ will be unique if we insist that the extension fields have conformal dimension $\tfrac{1}{2}$.  Indeed, this requires $\ell = 1$ and $n=0$.  We are therefore extending $\AKMSA{gl}{1}{1}$ by the fields associated with the atypical modules $\AffAtypMod{0,1}$ and $\AffAtypMod{0,-1}$.  Since $2n \ell = 0$ is even, the generators of the resulting extended algebra, $\alg{W}_{0,1}$, may be assigned a definite parity:  $\bc = V_{0,1}^+$ and $\bar{\bc} = V_{0,-1}^-$ are odd, $\bb = V_{0,1}^-$ and $\bg = -V_{0,-1}^+$ are even.  The expansions \eqref{eqnGenExtAlgOPEs} become
\begin{equation}
\begin{aligned}
\func{\bc}{z} \func{\bar{\bc}}{w} &= \frac{1}{z-w} + \func{N}{w} + \frac{1}{2k} \func{E}{w} + \ldots , \\
\func{\bb}{z} \func{\bg}{w} &= \frac{1}{z-w} + \func{N}{w} - \frac{1}{2k} \func{E}{w} + \ldots ,
\end{aligned}
\qquad
\begin{aligned}
\func{\bb}{z} \func{\bc}{w} = +\frac{\func{\psi^+}{w}}{\sqrt{k}} + \ldots , \\
\func{\bg}{z} \func{\bar{\bc}}{w} = -\frac{\func{\psi^-}{w}}{\sqrt{k}} + \ldots ,
\end{aligned}
\end{equation}
which we recognise as a free complex fermion $\brac{\bc,\bar{\bc}}$ and a $\beta \gamma$ ghost system.  Because the mixed \opes{} are regular, $\alg{W}_{0,1}$ decomposes into the direct sum of the chiral algebras of these theories.

If we choose to extend by dimension $1$ fields, then there are two distinct choices:  $n = \tfrac{1}{2}$ and $\ell = 1$ or $n = \tfrac{1}{2}$ and $\ell = -2$.  We expect a current algebra symmetry in both cases.  Indeed, if we set $\BH = N + E / \ell k$ and $\BZ = N - E / \ell k$, then we discover that the $\brac{\BH , \BZ}$-weights of the $\AKMSA{gl}{1}{1}$ currents and the extension fields $V_{n,\ell}^{\pm}$, $V_{-n,-\ell}^{\pm}$ precisely match the $\brac{\BH , \BZ}$-weights of the adjoint representation of $\SLSA{sl}{2}{1}$.\footnote{Here, $\BH$ and $\BZ$ should be associated with the matrices $\diag \set{1,-1,0}$ and $\diag \set{1,1,2}$ in the defining representation of $\SLSA{sl}{2}{1}$.}  Moreover, we have
\begin{equation}
\func{\BH}{z} \func{\BH}{w} = \frac{2 / \ell}{\brac{z-w}^2} + \ldots , \qquad \func{\BZ}{z} \func{\BZ}{w} = \frac{-2 / \ell}{\brac{z-w}^2} + \ldots ,
\end{equation}
and $\func{\BH}{z} \func{\BZ}{w}$ regular, which suggests that the extended algebra will be $\AKMSA{sl}{2}{1}$ at level $1 / \ell$.

Checking this for the choice $\ell = -2$ is easy.  As $2n \ell = -2$ is even, $\alg{W}_{1/2,-2}$ admits a superalgebra structure.  Moreover, the fusion rules
\begin{equation}
\AffAtypMod{0,1} \fuse \AffAtypMod{0,1} = \AffAtypMod{-1/2,2}, \qquad \AffAtypMod{0,-1} \fuse \AffAtypMod{0,-1} = \AffAtypMod{1/2,-2}
\end{equation}
imply that $\alg{W}_{1/2,-2}$ is a subalgebra of the extended algebra $\alg{W}_{0,1}$ considered above.  One readily checks that by taking normally-ordered products, the $\beta \gamma$ ghost fields of $\alg{W}_{0,1}$ generate the bosonic subalgebra $\AKMA{sl}{2}_{-1/2} \subset \AKMSA{sl}{2}{1}_{-1/2}$, the complex fermion gives the $\AKMA{u}{1}$-subalgebra, and the mixed products yield the remaining fermionic currents.  This establishes the superalgebra isomorphism $\alg{W}_{1/2,-2} \cong \AKMSA{sl}{2}{1}_{-1/2}$.

The computation when $\ell = 1$ is, however, more subtle because $2n \ell = 1$ is odd, so $\alg{W}_{1/2,1}$ does not admit the structure of a superalgebra.  To impose the correct parities on the extended algebra currents, we must adjoin an operator-valued function $\mu$ which is required to satisfy
\begin{equation} \label{eqCocycle}
\mu_{a,b} \mu_{c,d} = (-1)^{ad} \mu_{a+b,c+d}, \qquad \text{($a,b,c,d \in \ZZ$).}
\end{equation}
Note that the algebra generated by these operators has unit $\mu_{0,0}$.  The currents are then given by
\begin{equation}
\begin{aligned}
\BE &= +\mu_{1,1} V_{1/2,1}^+, \\
\BF &= -\mu_{-1,-1} V_{-1/2,-1}^-,
\end{aligned}
\qquad
\begin{aligned}
\BH &= N + E/k, \\
\BZ &= N - E/k,
\end{aligned}
\qquad
\begin{aligned}
\Be^+ &= -\mu_{1,0} \psi^+ / \sqrt{k}, \\
\Bf^- &= +\mu_{-1,0} \psi^- / \sqrt{k},
\end{aligned}
\qquad
\begin{aligned}
\Bf^+ &= \mu_{0,-1} V_{-1/2,-1}^+, \\
\Be^- &= \mu_{0,1} V_{1/2,1}^-,
\end{aligned}
\end{equation}
and routine computation now verifies that these currents indeed generate $\AKMSA{sl}{2}{1}_1$.

As our final example, we briefly consider the case of extensions of conformal dimension $\tfrac{3}{2}$.  There are now three distinct choices, corresponding to $n=1$, $\ell=1$, or $n=-\tfrac{1}{4}$, $\ell=2$, or $n=-1$, $\ell=3$.  The latter choice again results in an extended algebra which is a subalgebra of $\alg{W}_{0,1}$ because
\begin{equation}
\AffAtypMod{0,1} \fuse \AffAtypMod{0,1} \fuse \AffAtypMod{0,1} = \AffAtypMod{-1,3}.
\end{equation}
Both $\alg{W}_{1,1}$ and $\alg{W}_{-1,3}$ are superalgebras, while $\alg{W}_{-1/4,2}$ is not.  We expect, however, that a modification similar to \eqref{eqCocycle} will restore the superalgebra parity requirements.  We will not analyse this in any detail as our interest in $\Delta_{n,\ell} = \tfrac{3}{2}$ lies not with the full extended algebra, but rather with one of its subalgebras.

We start with the superalgebras $\alg{W}_{1,1}$ and $\alg{W}_{-1,3}$.  Both $V_{-n,-\ell}^+$ and $V_{n,\ell}^-$ are bosonic and upon defining
\begin{equation}
\begin{gathered}
\sg^+ = \sqrt{\frac{3 \alpha \brac{3 \alpha - 1}}{2 \mu_{\ell}^{\brac{0}}}} \: V_{-n,-\ell}^+, \qquad 
\sg^- = \sqrt{\frac{3 \alpha \brac{3 \alpha - 1}}{2 \mu_{\ell}^{\brac{0}}}} \: V_{n,\ell}^-, \\
\sj = -\alpha \partial X_{n-1/2,\ell}, \qquad 
\st = \frac{\alpha}{2} \normord{\partial X_{n-1/2,\ell} \partial X_{n-1/2,\ell}} - \frac{\ell \brac{\ell + 1}}{2} \frac{\alpha \brac{3 \alpha - 1}}{\alpha + 1} \frac{\normord{\psi^+ \psi^-}}{k},
\end{gathered}
\end{equation}
where
\begin{equation}
\alpha = \frac{1}{\brac{2n-1} \ell},
\end{equation}
we obtain the defining relations of the \emph{Bershadsky-Polyakov algebra} $W_3^{\brac{2}}$ \cite{PolGau90,BerCon91}:
\begin{equation}
\begin{gathered}
\func{\sg^+}{z} \func{\sg^-}{w} = \frac{\brac{K+1} \brac{2K+3}}{\brac{z-w}^3} + \frac{3 \brac{K+1} \func{\sj}{w}}{\brac{z-w}^2} + \frac{3 \func{\normord{\sj \sj}}{w} + \tfrac{3}{2} \brac{K+1} \func{\partial \sj}{w} - \brac{K+3} \func{\st}{w}}{z-w} + \ldots , \\
\func{\sj}{z} \func{\sg^{\pm}}{w} = \frac{\pm \func{\sg^{\pm}}{w}}{z-w} + \ldots , \qquad 
\func{\sj}{z} \func{\sj}{w} = \frac{\brac{2K+3}/3}{\brac{z-w}^2} + \ldots , \\
\func{\st}{z} \func{\sg^{\pm}}{w} = \frac{3}{2} \frac{\func{\sg^{\pm}}{w}}{\brac{z-w}^2} + \frac{\func{\partial \sg^{\pm}}{w}}{z-w} + \ldots , \qquad 
\func{\st}{z} \func{\sj}{w} = \frac{\func{\sj}{w}}{\brac{z-w}^2} + \frac{\func{\partial \sj}{w}}{z-w} + \ldots , \\
\func{\st}{z} \func{\st}{w} = \frac{-\brac{2K+3} \brac{3K+1} / 2 \brac{K+3}}{\brac{z-w}^4} + \frac{2 \func{\st}{w}}{\brac{z-w}^2} + \frac{\func{\partial \st}{w}}{z-w} + \ldots
\end{gathered}
\end{equation}
Here, the $\AKMA{sl}{3}$-level $K = \tfrac{3}{2} \brac{\alpha - 1}$ is $0$ for $\alg{W}_{1,1}$ and $-\tfrac{5}{3}$ for $\alg{W}_{-1,3}$.  The central charge of the $W_3^{\brac{2}}$-subalgebra is in both cases $-1$.

For $\alg{W}_{-1/4,2}$, this procedure does not yield a Bershadsky-Polyakov algebra because $V_{-n,-\ell}^+$ and $V_{n,\ell}^-$ are, in this case, mutually fermionic.  Rather, these fields generate a copy of the $\mathcal{N} = 2$ superconformal algebra of central charge $-1$.  Instead, we must consider the mutually bosonic fields $V_{n,\ell}^+$ and $V_{-n,-\ell}^-$.  Taking
\begin{equation}
\sg^{+} = \sqrt{3} \: V_{1/4,-2}^-, \quad \sg^{-} = \sqrt{3} \: V_{-1/4,2}^+, \quad  \sj = -\partial X_{1/4,2}, \quad \st = \frac{1}{2} \normord{\partial X_{1/4,2} \partial X_{1/4,2}} - \frac{1}{k} \normord{\psi^+ \psi^-}
\end{equation}
in particular, now leads to the Bershadsky-Polyakov algebra of level $0$ and central charge $-1$.  (In contrast, $V_{n,\ell}^+$ and $V_{-n,-\ell}^-$ are fermionic in both $\alg{W}_{1,1}$ and $\alg{W}_{-1,3}$, generating copies of the $\mathcal{N} = 2$ superconformal algebra with central charges $1$ and $-1$, respectively.)

\subsection{$W^{\brac{2}}_N$-subalgebras}

In the previous section, we found the Bershadsky-Polyakov algebra $W^{\brac{2}}_3$, at certain levels, appearing as a subalgebra of the extended algebras $\alg{W}_{1,1}$, $\alg{W}_{-1/4,2}$ and $\alg{W}_{-1,3}$.  We now generalise this observation.  The algebra $W^{\brac{2}}_3$ is defined \cite{PolGau90,BerCon91} as the Drinfel'd-Sokolov reduction of $\AKMA{sl}{3}$ corresponding to the non-principal embedding of $\SLA{sl}{2}$ in $\SLA{sl}{3}$.  Feigin and Semikhatov \cite{Feigin:2004wb} found that it could also be realised as a subalgebra of $\AKMSA{sl}{3}{1} \oplus \AKMA{u}{1}$ commuting with an $\AKMA{sl}{3}$-subalgebra.  They then studied a generalisation $W^{\brac{2}}_N \subset \AKMSA{sl}{N}{1} \oplus \AKMA{u}{1}$ which commutes with the obvious $\AKMA{sl}{N}$-subalgebra.

When $N=1$, these generalisations reduce to the chiral algebra of the $\beta \gamma$ ghost system.  For $N=2$, one gets $\AKMA{sl}{2}$, and as mentioned above, $N=3$ recovers the Bershadsky-Polyakov algebra.  The examples studied in \secref{secExamples} therefore lead us to the plausible conjecture that the $W^{\brac{2}}_N$ algebras of Feigin and Semikhatov may be realised, at least for certain levels, as subalgebras of certain of our extended algebras $\alg{W}_{n,\ell}$.  We mention that there is a second construction of these $W^{\brac{2}}_N$ algebras, but restricted to the critical level $K=-N$ (see \eqref{eq:W_NC}), starting from the affine superalgebra $\AKMSA{psl}{N}{N}$ at (critical) level $0$ \cite{CGL}.

Feigin and Semikhatov only computed the first few terms of the defining \opes{} of $W^{\brac{2}}_N$.  We will compare these terms with those obtained from our extended algebras, finding decidedly non-trivial agreement.  Our findings will, however, be stated as conjectures because the full \ope{} of $W^{\brac{2}}_N$ is not currently known.  $W^{\brac{2}}_N$ is generated by two fields $\CE^{\pm}_N$ of dimension $\tfrac{1}{2} N$, a $\AKMA{u}{1}$-current $\CH_N$ and an energy-momentum tensor $\CT_N$.  The defining expansions are:
\begin{equation}\label{eq:W2nope}
\begin{split}
\func{\CH_N}{z} \func{\CH_N}{w} &= \frac{\brac{N-1} K/N + N-2}{\brac{z-w}^2} + \ldots , \qquad
\func{\CH_N}{z} \func{\CE^{\pm}_N}{w} = \pm \frac{\func{\CE^{\pm}_N}{w}}{z-w} + \ldots , \\
\func{\CE^+_N}{z} \func{\CE^-_N}{w} &= \frac{\lambda_{N-1}}{\brac{z-w}^N} +  \frac{N \lambda_{N-2} \func{\CH_N}{w}}{\brac{z-w}^{N-1}} - \frac{\brac{K+N} \lambda_{N-3} \func{\CT_N}{w}}{\brac{z-w}^{N-2}} \\
&\mspace{-20mu} + \frac{\lambda_{N-3}}{\brac{z-w}^{N-2}} \sqbrac{\frac{N \brac{N-1}}{2} \func{\normord{\CH_N \CH_N}}{w} + \frac{N \bigl( \brac{N-2} \brac{K+N-1} - 1 \bigr)}{2} \func{\partial \CH_N}{w}} + \ldots
\end{split}
\end{equation}
Here, $\lambda_m = \prod_{i=1}^m \bigl( i \brac{K+N-1} - 1 \bigr)$, $K$ is the level of the $W^{\brac{2}}_N$ algebra, and the central charge is given by
\begin{equation} \label{eq:W_NC}
C = -\frac{\bigl( \brac{K+N} \brac{N-1} - N \bigr) \bigl( \brac{K+N} \brac{N-2} N - N^2 + 1 \bigr)}{K+N}.
\end{equation}

Suppose first that $2n\ell$ is even, so we can consider the bosonic subalgebra generated by the fields
\begin{equation}
\CE^+_N = \sqrt{\frac{\lambda_{N-1}}{\mu_{\ell}^{\brac{0}}}} \: V_{-n,-\ell}^+, \qquad \CE^-_N = \sqrt{\frac{\lambda_{N-1}}{\mu_{\ell}^{\brac{0}}}} \: V_{n,\ell}^-.
\end{equation}
Evaluating the \ope{} of these fields using \eqref{eqnGenExtAlgOPEs} and comparing with \eqref{eq:W2nope}, we find that the first two singular terms agree provided that $N = 2 \Delta_{n,\ell}$ and $\CH_N = -\partial X_{n-1/2,\ell} / \brac{2n-1} \ell$.  This also fixes the $W^{\brac{2}}_N$ level $K$.  Comparing the third terms fixes the form of the $W^{\brac{2}}_N$ energy-momentum tensor $\CT_N$ and $\CH_N$ is then verified to have dimension $1$.  However, the $\CE^{\pm}_N$ only have the required dimension $\tfrac{1}{2} N = \Delta_{n,\ell}$ if $n=1$ or $2n+\ell=1$.\footnote{There is a third solution, $\Delta_{n,\ell} + \ell + 1 = 0$, but this is invalid as we require $\ell, \Delta_{n,\ell} > 0$.}  These constraints also let us check that $\CT_N$ is an energy-momentum tensor and the central charge turns out to be $C=-1$.  When $2n\ell$ is odd, we instead consider the bosonic subalgebra generated by
\begin{equation}
\CE^+_N = \sqrt{\frac{\lambda_{N-1}}{\mu_{\ell-1}^{\brac{0}}}} V_{-n,-\ell}^-, \qquad \CE^-_N = \sqrt{\frac{\lambda_{N-1}}{\mu_{\ell-1}^{\brac{0}}}} \: V_{n,\ell}^+.
\end{equation}
A similar analysis reveals that this subalgebra agrees with $W^{\brac{2}}_N$ up to the first three terms in the \opes{} provided that $N = 2 \Delta_{n,\ell}$ and either $\ell = 1$ or $\ell = 2$.\footnote{Taking $n = -\tfrac{1}{2} \brac{\ell + 1}$ also satisfies these requirements, but then $2n\ell$ is necessarily even.  Moreover, there is again a solution of the form $\Delta_{n,\ell} - \ell + 1 = 0$, but it is easy to check that it leads to the wrong \ope{} of $\CT_N$ with itself.}  In the first case, $C=1$; in the second, $C=-1$.

We summarise our findings as follows:
\begin{conjecture}
The extended algebra $\alg{W}_{n,\ell}$ has a subalgebra isomorphic to $W^{\brac{2}}_N$ of level $K$ when:
\begin{itemize}
\item $\ell = 1$ and $n = 0,1,2,\ldots$  Then, $N = 2n + 1$ and $K = -2 \brac{n-1} \brac{2n+1} / \brac{2n-1}$.
\item $\ell = 1$ and $n = \tfrac{1}{2}, \tfrac{3}{2}, \tfrac{5}{2}, \ldots$  Then, $N = 2n + 1$ and $K = - \brac{2n^2 - 1} / n$.
\item $\ell = 2$ and $n = -\tfrac{3}{4}, -\tfrac{1}{4}, \tfrac{1}{4}, \ldots$  Then, $N = 4 \brac{n+1}$ and $K = -2 \brac{n+1} \brac{4n+1} / \brac{2n+1}$.
\item $n = -\tfrac{1}{2} \brac{\ell - 1}$ and $\ell = 1,2,3,\ldots$  Then, $N = \ell$ and $K = -\brac{\ell^2 - \ell - 1} / \ell$.
\end{itemize}
\end{conjecture}
\noindent Note that the examples considered in \secref{secExamples} exhaust the $W^{\brac{2}}_N$-subalgebras with $N \leqslant 3$ except for $\ell = 2$ and $n = -\tfrac{3}{4}$.  This latter case is excluded if one insists, as we did with \eqref{eqdim}, that the \ope{} of $\CE^{\pm}$ with itself is regular.  We mention that Feigin and Semikhatov actually computed the first \emph{four} terms of the $W^{\brac{2}}_N$ \opes{}, finding in the fourth term a Virasoro primary field $\CW_N$ of dimension $3$ and $\CH_N$-weight $0$.  We have extended \eqnTref{eqnVertexOPE}{eqnCompositeSFOPEs}{eqnGenExtAlgOPEs} to compute $\CW_N$ in our extended algebras and have checked that for each $\ell$ and $n$ appearing in our conjecture, this field indeed has the required properties.  It follows that our conjecture has been verified for all $N \leqslant 4$.

%\bibliography{gl11}
%\bibliographystyle{unsrt}

\end{document}